\documentclass[conference]{IEEEtran}
\usepackage{psfrag}

\ifCLASSINFOpdf
\else
\fi
\usepackage[cmex10]{amsmath}

\usepackage{mdwmath}
\usepackage{mdwtab}
\usepackage{algorithm}
\usepackage{algorithmic}
\usepackage{graphicx,amsmath,amssymb,latexsym,subfigure}
\usepackage{psfrag}
\newtheorem{theorem}{Theorem}

\begin{document}
%

\title{Finite Horizon Online Lazy Scheduling with Energy Harvesting Transmitters over Fading Channels}


\author{\IEEEauthorblockN{Baran Tan Bacinoglu and Elif Uysal-Biyikoglu}
\IEEEauthorblockA{Dept. of Electrical and Electronics Eng., METU\\
Ankara, TURKEY 06800\\
Telephone: +90 (312) 210--5764\\
e--mail: tbacinoglu@gmail.com, elif@eee.metu.edu.tr}}


%


\maketitle

\begin{abstract}

Lazy scheduling, i.e. setting transmit power and rate in response to data traffic as low as possible so as to satisfy delay constraints, is a known method for energy efficient transmission. This paper addresses an online lazy scheduling problem over finite time-slotted transmission window and introduces low-complexity heuristics which attain near-optimal performance. Particularly, this paper generalizes lazy scheduling problem for energy harvesting systems to deal with packet arrival, energy harvesting and time-varying channel processes simultaneously. The time-slotted formulation of the problem and depiction of its offline optimal solution provide explicit expressions allowing to derive good online policies and algorithms.     
\end{abstract}


%
\IEEEpeerreviewmaketitle

\section{Introduction}

\def\eg{{e.g.}}
\def\ie{{i.e.}}

There have been offline and online problem formulations for energy efficient packet scheduling with data arrival and deadline constraints (e.g., \cite{UPE02,BeGa02,NuSr02,ZaMo09}) as well as intermittent energy availability constraints e.g., \cite{YaU2010,TuYe2010, MAAEUHE2010, Shroff2011, Tassiulas2010,Ozel2011}. Offline policies make their computations with complete prior knowledge of data/energy/channel variations, which is rarely a practical assumption, while providing performance benchmarks and an understanding of the structure of optimal rate/power adaptation. One criticism that offline formulations regularly received is that the resulting offline policies did little to suggest good online policies. On the other hand, direct online formulations have been disconnected from offline formulations and the resulting policies (optimal policies or heuristics) have lacked the appealing structure of offline policies.

This paper presents an approach for going from an offline formulation to an online policy. The problem posed (presented in section\ref{sec:general}) embodies data and energy causality constraints, as well as channel variations in a discrete time formulation, such that there is a finite number of time slots in which the data needs to be sent. Such a finite-horizon formulation is not only realistic considering practical scenarios, but also makes a substantial difference in in the nature of the resulting policies (see, e.g. ~\cite{HoZang2010} in contrast with ~\cite{Sharma}.) As observed in ~\cite{Shroff2011}, optimality in the finite horizon, as opposed to infinite horizon, requires a much finer control. Depending on the short-term statistics of channel, energy and traffic, a long-term throughput optimal (i.e. stable) policy (e.g., \cite{Sharma}) may be quite far off in terms of finite-horizon throughput or energy consumption, compared with even a simple adaptive heuristic ~\cite{OurJournalSubmission}.

The rest of the paper proceeds in two main parts. First,  in section \ref{sec:onlinelazyprobdef}, a basic online {\emph{lazy scheduling}}~\cite{PUE01} problem is posed and solved. In this basic problem, the goal is to schedule packets arriving in a finite time window while minimizing the energy cost. Secondly, a more general formulation is made in Section \ref{sec:general} with energy harvests and channel state variation are included on top of data arrivals. This leverages the solution of the lazy scheduling problem and the Expected Threshold scheduling heuristic developed in \ref{sec:heuristic}. The proposed policies are tested with respect to online optimal adaptation on simulations, followed by discussion and conclusions.

\section{Online Lazy Scheduling}
\label{sec:onlinelazyprobdef}
Consider a network node receiving arbitrary amounts of data that need to be transmitted to a further destination within a finite time interval. The problem is to adjust the outgoing transmission rate (jointly, transmit power) in time in response to the incoming data traffic, to minimize the expected total energy consumption. Let us define a slot duration as the smallest interval of time in between two adjustments of rate/power. 

We will consider a finite period of $n$ slots in which the transmission needs to be completed.  Of course, delivering all received bits within this period may not be guaranteed by some policies. To account for this, a cost $C\left(b\right)$ is assigned for retaining a backlog of $b$ bits at the end of the time horizon. Naturally, $C\left(b\right)$ is monotone nondecreasing and equals zero at $b=0$.

Let the data arrival process $\lbrace B_n \rbrace, n\geq 1$ be a discrete time Markov process, such that $B_n\geq 0$ is the size of the data packet (in bits) received at the beginning of slot $n$.  The energy used during one slot for transmitting at rate $r$ (bits/slot) will be given by the following: 
\begin{equation}
 e(b,r)= p(r)\min\left( \frac{b}{r},1\right) 
\label{eq:Wr}
\end{equation}
The power function, $p(r)$, will be assumed to be convex and increasing in $r$. The minimum function takes care of the case when $b$, the number of bits stored in the in the buffer at the beginning of the slot, is less than $r$, in which case the transmission will cover only part of the slot.

Let $i$ be the packet arrival state and $l(i)$ be the function that returns the corresponding packet length when the packet arrival state is $i$. Then, the state of the system at time $n$ is determined by the vector $(b(n),i(n))$. No buffer limits are imposed, meaning the sender can store arbitrary amounts of data. When the system state is $(b,i)$, let $J_n \left( b,i\right) $ be the minimum expected total energy consumption from the current time, $n$, until the end of the time horizon. Then, the problem can be formulated using 
a stochastic dynamic programming equation as below.
 \begin{equation}
 J_n \left( b,i\right)= \displaystyle\min_{r}\left[e(b,r)+ \displaystyle\sum_{j}A_{ij}J_{n+1} \left( (b-r)_{+}+l(j),j\right)\right] 
\label{eq:Jn}
\end{equation}
  
where $A_{ij}$ is the transition probability from packet arrival state $i$ to state $j$. The optimal solution (that minimizes the expected energy consumption) is given by:

\begin{equation}
r_{n}^{*}= \displaystyle\arg\min_{r}\left[e(b,r)+ \displaystyle\sum_{j}A_{ij}J_{n+1} \left( (b-r)_{+}+l(j),j\right)\right] 
\label{eq:rn}
\end{equation}
 
The cost function $J_n \left( b,i\right) $ and minimizing rates $r_{n}^{*}$'s can be computed by backward induction starting from the last slot. The function  $J_{N+1} \left( b,i\right) $ can be interpreted as a penalty function which corresponds to the cost of maintaining bits of packets not delivered until the deadline. Accordingly, $J_{N+1} \left( b,i\right) = C\left(b-l(i)\right)$.

To model the natural restriction in practical systems where the choice of transmission rates are limited, the rate  $r_{n}$ will be assumed to belong to a discrete set $\mathbf{V}$.

The optimal solution of the dynamic programming formulation in Eq. \ref{eq:Jn}, obtained by backward induction, has exponential complexity over the product of $\mathbf{V}$ and state space, such that, the energy required for computing such a solution may well exceed the energy savings it was designed to achieve. In the following we exhibit a simple yet efficient policy, that we call the "Expected Threshold Lazy Scheduling Policy", performing close to the dynamic programming optimal solution.

\section{Expected Threshold Lazy Scheduling Policy}
\label{sec:etls} 
Let us start by considering the offline solution where transmission rates are not restricted to a discrete set. The \emph{stretched string} method \cite{ZaMo09} can be employed to find optimal offline transmission rates since the optimal departure curve follows the shortest path. Depending on the cost function $C\left(b\right)$, the optimal offline solution does not have to be the one that completes transmission of all incoming data by the end of the time horizon, instead, it may retain a certain amount in the buffer in exchange for minimizing the energy consumption on the rest of the data. That case is effectively captured by the time horizon virtually extended to the point where the data buffer is emptied with the transmission rate selected at the last time slot. (As there are no arrivals after $T$, there is no reason to change the rate from the beginning of the last slot until the end of the extended period. Optimality of keeping a constant rate follows from convexity of $p(r)$.) 

For a given cost function, the amount of extra time $\alpha$ will be as a result of the offline solution. In particular, if the cost function was in the form

\begin{equation}
C\left(b\right)= \tau p(\frac{b}{\tau})
\end{equation}  

Then, the extension $\alpha$ is exactly equal to $\tau$. In general,  the optimal offline rate $\tilde{r}_{n}^{*}$ can be expressed as a function of $\alpha$ as below:

\[
\tilde{r}_{n}^{*}= \displaystyle\min_{a=1,....,(N+\alpha-n)} \left( b_{n},\tilde{r}_{n}(b_{n},a)\right) ,\text{  where }
\]
\[
\tilde{r}_{n}(b_{n},a)=[b_{n}+\displaystyle\sum_{l=n+1}^{n+a} B_{l}]/a
\]

where $B_{n}$ represents the size of the packet arriving becoming available at the beginning of slot $n$ ($B_l=0$, $l>t$) The case $\alpha\to 0$  corresponds to  the cost function $C\left(b\right)\to \infty$ for any value of $b>0$. An online lazy schedule can be constructed by setting transmission rate to the expectation of the offline transmission rate given above. For simplicity, $E[\tilde{r}_{n}^{*}]$ can be approximated and the following expression for online decisions can be derived:

\begin{equation}
r_{n}= \displaystyle\min \left\lbrace \rho \in \mathbf{V} \vert \rho > \displaystyle\max(b_{n},E\left[ \tilde{r}_{n}(b_{n},1) \vert B_{n}^{N}\right] \right\rbrace
\label{rebolapa}  
\end{equation}

Alternatively, considering the stretched string visualization:
\begin{eqnarray*}
\\r_{n} &
=&
 \displaystyle\min \left\lbrace \rho \in \mathbf{V} \vert L_{n}(B_{n}^{N},r) \geq b_{n} \right\rbrace
\\ L_{n}(B_{n}^{N},r)&=& \displaystyle\max\left( r,rn+r\alpha - \displaystyle\sum_{l=1}^{n-1}E\left[B_{l} \vert B_{n}^{N}\right]\right) 
\\\mbox{for $r \neq r_{max}$}&;& L_{n}(B_{n}^{N},r_{max})=\infty 
\end{eqnarray*}

where $B_{n}^{N}$  is the vector of packet sizes $[B_{n}...B_{N}]$ and $L_{n}(B_{n}^{N},r)$ is the data buffer threshold for selecting transmission rate $r$. In the following, we include several simple and basic suboptimal solutions for comparison purposes.


\par{{\textbf{Hasty Policy:}}}
This is essentially a greedy policy. Contrary to the conservative policy above, it always selects the highest possible transmission rate $r_{max}$ in the set $\mathbf{V}$ that will not cause idleness in the next slot. The rate allocation is $r_{\rm{hasty}}(b)= \max \left\lbrace r \in \mathbf{V} \vert r < b\right\rbrace$. This policy is likely to perform well for a steep cost function $C\left(b\right)$, but of course, as it will be unnecessarily hasty at times, it is still suboptimal.

\par{{\textbf{Constant Rate Policy:}}}
This policy uses a single transmission rate that keeps the data buffer stable. In particular,  it selects the lowest transmission rate in the set $\mathbf{V}$ which is above the average data arrival rate. In the case where the arrival rate approaches the chosen constant rate (from below), the policy is asymptotically optimal, i.e. throughput maximizing. Yet, as will be shown in simulations, it may be far from optimal in the short term.

\section{Generalized Online Lazy Scheduling}
\label{sec:general} 
The online lazy scheduling problem defined in the previous section does not consider energy harvests or channel variation.  A generalization will now be made on this problem by considering an energy arrival process and the possibility of energy depletion during transmission, as well as channel states changing from slot to slot. 

\subsection{Problem Definition}
Let $\lbrace H_n \rbrace$, $\lbrace B_n \rbrace$ and $\lbrace \gamma_n \rbrace$ be discrete time Markov processes representing energy arrivals, packet arrivals and channel fading, respectively, where $n$ is the time slot index. Particularly, $H_n$ is the amount energy that becomes available in slot $n$ (harvested during slot $n-1$), $B_n$ is the size of the packet that becomes available at the beginning of slot $n$ and $\gamma_n$ is the average channel gain level in slot $n$. As in the online lazy problem, the objective of this problem is to minimize a cost which is the sum of the energy cost of undelivered data and the total energy consumed within a transmission time window of $N$ slots.

By the Markovian assumption, in order to make the optimal decision for transmission power and rate, it is sufficient to know the present battery level, data buffer state and channel gain. Let $(e_n, b_n,\gamma_n)$ be the state vector representing these values on time slot $n$.  The following dynamic programming equation relates the cost $J_n(e_n,b_n,\gamma_n)$ of being in this state at time $n$ to the cost for the next time slot $n+1$ where $s$ denotes the length of a slot:

\begin{eqnarray*}
&J_n(e_n,b_n,\gamma_n)=& \\
&\displaystyle\min_{(\rho_{n},r_{n})\in \mathbf{M}} [ \dot{s}_{n}\rho_{n}+& \\
&E[J_{n+1}( e_{n} - \! \dot{s}_{n}\rho_{n}\! +\! H_{n+1},\! b_{n}\! -\dot{s}_{n}r_{n}+B_{n+1},\gamma_{n+1})]  ]  &
\end{eqnarray*}

where 
\[
\dot{s}_{n}=s\min(\frac{e_{n}}{s\rho_{n}},\frac{b_{n}}{sr_{n}},1)
\]
$\rho_{n}$ and $r_{n}$ are transmission power and rate decisions which are chosen from a finite set $\mathbf{M}$.

The cost at the last time slot of transmission window is given by:
\begin{equation}
J_N(e_N,b_N,\gamma_N)= \displaystyle\min_{(\rho_N,r_N)\in \mathbf{M}} \left[  \dot{s}_{N}\rho_{n}+E[C(b_{N}-\dot{s}_{N}r_{N})]\right]
\end{equation}

For a given transmision power level $\rho_{n}$, transmission rate $r_{n}$ is restricted to a maximum value of transmission rate which is determined by a certain bit error rate and channel gain level. Therefore, the set $\mathbf{M}$ is also a function of channel state $\gamma_n$.

Transmission rate $r_{n}$ may be a function of signal to noise ratio (SNR) which is essentially proportional to the product of transmission power and channel gain. ($r_{n}=g(\rho_{n}\gamma_n)$)

Accordingly, the optimal decision for a time slot $n$ can be expressed by a pair of transmission power and rate $(\rho_{n},r_{n})$ or only one of them if they are one-to-one related.

In its most general form, the dynamic programming formulation of the problem suffers from being high dimensional since it requires states to be evaluated individually. For this reason, rather than inspecting dynamic programming solution, we take an alternative approach and introduce an online heuristic solution benefiting from optimal offline solution of the same problem as it is done in section \ref{sec:etls} for ETLS policy.

\subsection{Offline Solution}

The offline solution derived here covers a particular version of the problem making two basic assumptions: The transmission power and rate have a one-to-one relation through the AWGN channel capacity formula and the energy cost of retaining $b$ amount of data $C(b)$ is infinite for any nonzero value of $b$. 

In the rest, let $r_{n}$ be equal to the AWGN capacity of the channel such that $r_{n}=W\log_{2}(1+\rho_{n}\gamma_n)$ where $W$ is bandwidth of the channel. Then, energy and packet arrival contraints for a time slot $n$ can be expressed by at most $2T-2n+2$ inequalities. ($N-n+1$ for energy arrivals and $N-n+1$ for packet arrivals.)

\begin{equation}
s\displaystyle\sum_{l=n}^{n+u}\rho_{l} \leq e_{n}+ \displaystyle\sum_{l=n+1}^{n+u} H_{l}, u=1,2,.....,(N-n),
\end{equation}
\[
s\rho_{n} \leq e_{n}
\]
\begin{equation}
\displaystyle\sum_{l=n}^{n+v} sW \log_{2}(1+\rho_{l}\gamma_{l})\leq b_{n}+ \displaystyle\sum_{l=n+1}^{n+v} B_{l}, v=0,1,2,.....,(N-n)
\end{equation}
\[
 sW \log_{2}(1+\rho_{n}\gamma_{n})\leq b_{n}
\]
where $s$ is the length of a time slot.

Let the transmission power decision $\rho_{n}$ be determined by a {\emph{water level}} $w_{n}$ so that $\rho_{n}=(w_{n}-\frac{1}{\gamma_{n}})_{+}$. (This is required to minimize energy consumption per transmitted data under given channel constraints, see, \eg, \cite{HoZang2010}.) Then, the above inequalities can be rewritten as in below:

\begin{equation}
\displaystyle\sum_{l=n}^{n+u}s(w_{l}-\frac{1}{\gamma_{l}})_{+} \leq e_{n}+ \displaystyle\sum_{l=n+1}^{n+u} H_{l},
\end{equation}
\[
 u=0,1,2,.....,(N-n)
\]
\begin{equation}
\displaystyle\sum_{l=n}^{n+v} sW \log_{2}(1+(w_{l}-\frac{1}{\gamma_{l}})_{+}\gamma_{l})\leq b_{n}+ \displaystyle\sum_{l=n+1}^{n+v} B_{l}, 
\end{equation}
\[
v=0,1,2,.....,(N-n)
\]
The water level $w_{n}$ should be nondecreasing in time ($w_{n}\leq w_{n+1}$) because otherwise one can always reallocate consumed energy and transmitted bits to improve overall energy/bit efficiency without violating causality constraints due to energy and packet arrivals. 

\begin{theorem}
In an optimal offline transmission schedule, the water level $w_{n}$ is non-decreasing with slot index $n$.
\end{theorem}


\begin{proof}
We will show that if the water level of any slot $n$ is higher than the water level of the next slot $n+1$ ($w_{n} > w_{n+1}$), then, there is an offline transmision schedule which achieves at least the same throughput or consumes at the most the same amount of energy with the initial schedule. Consider the energy allocation and total throughput obtained in the period consisting of slots $n$ and $n+1$. Since the slot $n$ is the predecessor of the slot $n+1$, energy consumed or data transmitted within the slot $n$ can be transferred to  the slot $n+1$. 
The total throughput for slot $n$ and $n+1$ is equal to the following expression:

\[
\left( \log_{2}(w_{n})-\log_{2}(\frac{1}{\gamma_{n}})\right)_{+}+\left( \log_{2}(w_{n+1})-\log_{2}(\frac{1}{\gamma_{n+1}})\right)_{+}
\]

which can be maximized when $w_{n}=w_{n+1}$ if the total consumed energy for slot $n$ and $n+1$ ($( w_{n}-\frac{1}{\gamma_{n}}) _{+}+( w_{n+1}-\frac{1}{\gamma_{n+1}})_{+}$) is fixed. Similary, if the total throughput for slot $n$ and $n+1$ is fixed, the total consumed energy for slot $n$ and $n+1$ can  minimized by setting $w_{n}$ and $w_{n+1}$ to a common level. Therefore, if $w_{n} > w_{n+1}$,  reassigning water levels so that $w_{n} = w_{n+1}$ by keeping consumed energy or transmitted data constant does not decrease the total throughput or increase the total consumed energy amount. 
\end{proof}

Accordingly, the water level $w_{n}$ is 
bounded by following inequalities: 

\begin{equation}
\displaystyle\sum_{l=n}^{n+u}s(w_{n}-\frac{1}{\gamma_{l}})_{+} \leq e_{n}+ \displaystyle\sum_{l=n+1}^{n+u} H_{l}
\end{equation}
\[
 u=0,1,2,.....,(N-n)
\]
\begin{equation}
\displaystyle\sum_{l=n}^{n+v} sW \log_{2}(1+(w_{n}-\frac{1}{\gamma_{l}})_{+}\gamma_{l})\leq b_{n}+ \displaystyle\sum_{l=n+1}^{n+v} B_{l}
\end{equation}
\[
v=0,1,2,.....,(N-n)
\]

The above inequalities can be rearranged as in the following:
\begin{footnotesize}
\begin{equation}
  w_{n} \leq \frac{e_{n}+ \displaystyle\sum_{l=n}^{n+u} H_{l}+s\displaystyle\sum_{l=n+1}^{n+u}  M_{l}^{(e)}(w_{n}) }{s(u+1)}
\end{equation}
\end{footnotesize}
\[
 u=0,1,2,.....,(N-n)
\]
\vspace{-0.099in}
\begin{footnotesize}
\begin{equation}
\log_{2}(w_{n}) \leq \frac{b_{n}+ \!\!\!\!\!\displaystyle\sum_{l=n+1}^{n+v}\! \! B_{l}+sW\displaystyle\sum_{l=n}^{n+v} M_{l}^{(b)}}{sW(v+1)}
\end{equation}
\end{footnotesize}
\[
v=0,1,2,.....,(N-n)
\]
where 
\[
 M_{l}^{(e)}(w_{n})= \min (\frac{1}{\gamma_{l}},w_{n}), M_{l}^{(b)}= \log_{2}\left( \min (\frac{1}{\gamma_{l}},w_{n})\right)
\]
Therefore, the upper bound for $w_{n}$ can be expressed as 
$w_{n}\leq \min (w_{n}^{e},w_{n}^{b})$, where

\vspace{-0.099in}
\vspace{-0.099in}
\begin{footnotesize}
\begin{equation}
w_{n}^{e} =\displaystyle\min_{u=0,...,(N-n)} \frac{e_{n}+ \displaystyle\sum_{l=n+1}^{n+u} H_{l}+s\displaystyle\sum_{l=n}^{n+u}  M_{l}^{(e)}(w_{n}) }{s(u+1)}
\label{eq:eqwe}   
\end{equation}
\end{footnotesize}
\begin{footnotesize}
\begin{equation}
log_{2}(w_{n}^{b})=\displaystyle\min_{v=0,...,(N-n)} \frac{b_{n}+ \!\!\!\!\!\displaystyle\sum_{l=n+1}^{n+v} \!\!\!\!\!\ B_{l} +sW\displaystyle\sum_{l=n}^{n+v}  M_{l}^{(b)} }{sW(v+1)}   
\label{eq:eqwb} 
\end{equation}
\end{footnotesize}
As there is no other constraint on the water level $w_{n}$, it can be set to $\min (w_{n}^{e},w_{n}^{b})$. Thus,  the throughput maximizing water level is $w_{n}^{*}=\min (w_{n}^{e},w_{n}^{b})$.

If we assume $C(b)$ is infinite for $b> 0$,  all received data should be transmitted to have a finite total energy cost in generalized online lazy scheduling. The following theorem states that the offline optimal solution is a throughput maximizing schedule for this case.  

\begin{theorem}
Consider the case when $C(b_{N+1})=\infty$ for $b_{N+1}> 0$ and there exists a feasible offline solution that transmits all data within the time horizon. Then, throughput maximizing schedule with nondecreasing water levels also minimizes the total energy cost.
\end{theorem}   

\begin{proof}
To decrease energy consumption of a throughput maximizing schedule, it is needed to decrease water level for at least one slot where transmission power is nonzero (i.e. $w_{n}^{*}\geq\frac{1}{\gamma_{n}}$) but this also decreases the total throughput and  makes $b_{N+1}$ nonzero. Accordingly, in order to compensate the decrease  in the total throughput, water level of another slot (where transmission power is nonzero) should be increased and this is not possible since water levels are already set to their maximum value satisfying the energy-efficiency constraint that dictates $w_{n}\leq w_{n+1}$ for any time slot $n$. 
\end{proof}

The theorem also holds if the ratio $C(b)/b$ is larger than energy/bit rate during any time slot in all possible transmission schedules which guarantees that any decrease on water levels increases the total energy cost.Hence, offline optimal solutions of generalized online lazy scheduling problem are also throughput maximizing schedules where the cost $C(b)$ is sufficiently large.

As it can been seen in Eq. (\ref{eq:eqwe}) and Eq. (\ref{eq:eqwb}), $\min (w_{n}^{e},w_{n}^{b})$ is a nondecreasing function of $w_{n}$ and converges to a certain value as $w_{n}$ goes to infinity. 

The throughput maximizing water level $w_{n}^{*}$ can be found by iteratiratively evaluating $\min (w_{n}^{e},w_{n}^{b})$:
\begin{equation}
w_{n}^{(k+1)}=\vert_{w_{n}=w_{n}^{(k)}}\min (w_{n}^{e},w_{n}^{b})
\end{equation}
where $w_{n}^{(k)}$ is the $k$th iteration value and $w_{n}^{(1)}=w^{max}_{n}$.

The offline optimal power level $\rho_{n}^{*}$ that maximizes total throughput can be approached by the estimated water level ,$w_{n}^{(k)}$, which gives the optimal water level after infinitely many iteration steps.

\begin{equation}
\rho_{n}^{*}=\lim_{k \rightarrow \infty } (w_{n}^{(k)}-\frac{1}{\gamma_{n}})_{+}
\end{equation}

In practice, a few steps of iteration can be sufficient to obtain estimated water levels which are reasonably close to optimal water levels.  

\subsection{An Online Heuristic}
\label{sec:heuristic}
An online heuristic, which does not assume any prior knowledge of arrival process statistics, can be derived based on the throughput maximizing offline solution. For such a heuristic, the values of $w_{n}^{e}$ and $w_{n}^{b}$ can be estimated as follows:
\begin{equation}
\hat{w}_{n}^{e}= \left\{ \begin{array}{ll}
       \frac{e_{n}-\bar{H}_{n}}{s(N-n)}+\frac{\bar{H}_{n}}{s}+\bar{M}_{n}^{(e)}(w_{n})  & \mbox{; $e_{n} \geq \bar{H}_{n}$}\\
       \frac{e_{n}}{s}+\bar{M}_{n}^{(e)}(w_{n})& \mbox{; o.w.  }\\
        \end{array} \right. 
\end{equation}
\begin{equation}
\log_{2}(\hat{w}_{n}^{b})= \left\{ \begin{array}{ll}
       \frac{b_{n}-\bar{B}_{n}}{sW(N-n)}+\frac{\bar{B}_{n}}{s}+\bar{M}_{n}^{(b)}(w_{n})  & \mbox{; $b_{n} \geq \bar{B}_{n}$}\\
       \frac{b_{n}}{sW}+\bar{M}_{n}^{(b)}(w_{n})& \mbox{; o.w.  }\\
        \end{array} \right. 
\end{equation}
where
\vspace{-0.099in}
\[
\bar{H}_{n}=\frac{1}{n}\displaystyle\sum_{l=1}^{n} H_{l}, \bar{B}_{n}=\frac{1}{n}\displaystyle\sum_{l=1}^{n} B_{l}
\]
\vspace{-0.099in}
\[
\bar{M}_{n}^{(e)}(w_{n})=\frac{1}{n}\displaystyle\sum_{l=1}^{n} M_{l}^{(e)}(w_{n}),\bar{M}_{n}^{(b)}(w_{n})=\frac{1}{n}\displaystyle\sum_{l=1}^{n} M_{l}^{(b)}(w_{n})
\]

%
%
%
%
%
%
%

Then, the estimated value of throughput maximizing water level can be computed iteratively:

\begin{equation}
\hat{w}_{n}^{(k+1)}=\vert_{w_{n}=\hat{w}_{n}^{(k)}}\min (\hat{w}_{n}^{e},\hat{w}_{n}^{b})
\end{equation}

where $\hat{w}_{n}^{(k)}$ is the $k$th iteration of the estimated value of throughput maximizing water level and $\hat{w}_{n}^{(1)}=\min (\frac{e_{n}}{s},\frac{b_{n}}{sW})$.

In general, offline optimal water levels may remain constant for long time periods whereas estimated water levels $\hat{w}_{n}^{(k)}$s exhibit fluctuations. Accordingly, to further improve the heuristic estimation, exponential smoothing can be applied on $\hat{w}_{n}^{(k)}$s as in below.

\begin{equation}
\hat{v}_{n}^{(k)}=\beta\hat{w}_{n}^{(k)}+(1-\beta)\hat{v}_{n-1}^{(k)}
\end{equation}  

The smoothened value $\hat{v}_{n}^{(k)}$ can be used to decide on transmission power.

\begin{equation}
\rho_{n}^{heuristic}=(\hat{v}_{n}^{(k)}-\frac{1}{\gamma_{n}})_{+}
\end{equation}

\section{Numerical Study of the Policies}

\subsection{Online Lazy Scheduling}

A simulation experiment is performed to evaluate the expected threshold lazy scheduling (ETLS) policy against optimal policy using dynamic programming. The hasty and constant policies are also included for comparison. For the packet arrival process, a Markov model having two states ($l(0)=0$ (i.e. no packet arrival) and $l(1)=10$kB (packet arrival of constant size 10 KB)) with transition probabilities $q_{00}=0.9$, $q_{01}=0.1$, $q_{10}=0.58$, $q_{11}=0.42$ where slot duration is $1$ms. The set of data rates $\mathbf{V}$ is based on rates specified in the 802.11g standard (specifically, $6, 9, 12, 18, 24, 36, 48, 54$ Mbit/s.) In the computation of power levels corresponding to standard data rates,  the net data rate is assumed equal to Shannon capacity of an additive white Gaussian noise (AWGN) channel with a noise spectral density $0.83$ nW/Hz and $20$MHz bandwidth. (The average rate of incoming rate is just below $12$ Mbit/s.)The cost function $C\left(b\right)$ is chosen as $3 p(\frac{b}{3})$ and the extension parameter $\alpha$ of ETLS policy is set to $3$.

The performances of the policies are compared to the
optimal online policy considering both total energy consumed,
and the percentage of data they retain by the end of the horizon,
i.e. the percentage of received data that they fail to transmit
by $N$.

Fig.\ref{lazye} and Fig.\ref{lazyb} show the total energy consumption and
percentage of backlogged data of optimal and suboptimal
policies for individual realizations of packet arrival process.

\begin{figure}[htpb]
\centering
  \begin{psfrags}
    \psfrag{A}[t]{Realization Index}
	\psfrag{B}[b]{Energy Consumption (in mJ)}
	\psfrag{H}[l]{\tiny{HASTY}}
	\psfrag{C}[l]{\tiny{CONSTANT}}
	\psfrag{O}[l]{\tiny{OPTIMAL}}
	\psfrag{E}[l]{\tiny{ET POLICY}}
    \includegraphics[scale=0.27]{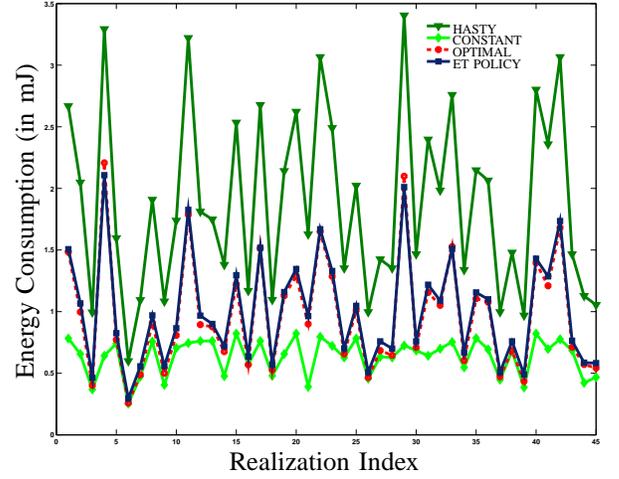}
    \end{psfrags}

\caption{Energy consumption comparison of Hasty, Constant rate policy, Optimal and Expected Threshold(ET) policies within a transmission window of $100$ slots over individual realizations of a Markovian stream of $10$kB packets  having two states ($l(0)=0$ and $l(1)=10$kB) with transition probabilities $q_{00}=0.9$, $q_{01}=0.1$, $q_{10}=0.58$, $q_{11}=0.42$. }
\label{lazye} 
\end{figure}

\begin{figure}[htpb]
\centering
  \begin{psfrags}
    \psfrag{A}[t]{Realization Index}
	\psfrag{B}[b]{Percentage of Backlogged Data}
	\psfrag{H}[l]{\tiny{HASTY}}
	\psfrag{C}[l]{\tiny{CONSTANT}}
	\psfrag{O}[l]{\tiny{OPTIMAL}}
	\psfrag{E}[l]{\tiny{ET POLICY}}
    \includegraphics[scale=0.27]{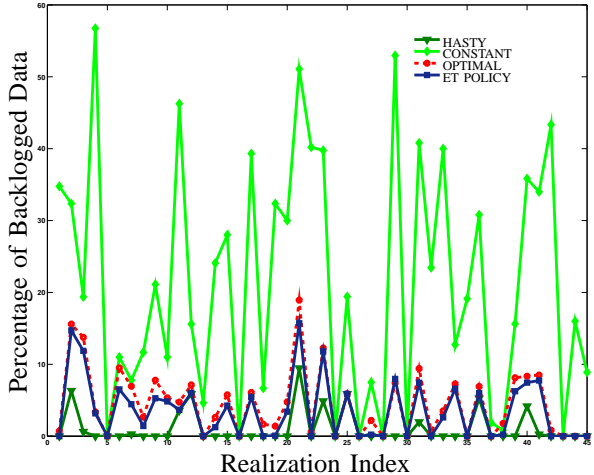}
    \end{psfrags}
\caption{Percentage of backlogged data for Hasty, Constant rate policy, Optimal and Expected Threshold(ET) policies within a transmission window of $100$ slots over individual realizations of a Markovian stream of $10$kB packets  having two states ($l(0)=0$ and $l(1)=10$kB) with transition probabilities $q_{00}=0.9$, $q_{01}=0.1$, $q_{10}=0.58$, $q_{11}=0.42$ .}
\label{lazyb} 
\end{figure}

%
%

\subsection{Generalized Online Lazy Scheduling}

In the second part of numerical study, the throughput
performances of throughput maximizing optimal offline policy
and online heuristic policy are compared.

In this simulation, rates are not restricted to a discrete set. The transmission window is $T=100$ slots where each slot has  duration $1$ms. The same packet arrival model is employed
as described above. Gilbert-Elliot channel is assumed where good ($\gamma^{good}=30$) and bad ($\gamma^{bad}=12$) states appear with equal probabilies. ($P(\gamma_{n}=\gamma^{good})=0.5$ and $P(\gamma_{n}=\gamma^{bad})=0.5$ ) Similarly, in energy harvesting process, energy harvests of $50$nJs are assumed to occur with a probability of $0.5$ at each slot. 


In Fig. \ref{wthroughput}, the throughput performance of the online heuristic policy is compared with the throughput maximizing optimal offline policy for individual realizations. 

A typical realization of packet arrival, energy harvesting and channel fading processes, and corresponding water level profiles is demostrated in Fig. \ref{sample12a}. As can be observed in Fig. \ref{sample12a}, the water levels of optimal offline policy make jumps to higher levels when both energy and packet arrivals have high intensities.

For another sample realization of packet arrival, energy harvesting and channel fading processes, water level profiles of throughput maximizing optimal offline policy and online heuristic policy are shown in Fig. \ref{sample12} (a) and (b). Fig. \ref{sample12} (a) shows water level profiles when transmission window size
$N$ is set to $100$ slots and Fig. \ref{sample12} (b) shows water level profiles when transmission window size is extended to $200$ slots. In the first $100$ slot, water level profiles are similar to each other though ,due to the relaxation of the deadline constraint, both optimal and heuristic water levels sligthly decrease when transmission window size is doubled. 

To illustrate the effect of transmission window size, average throughput performances and energy consumption of throughput maximizing offline optimal policy and online heuristic are compared against varying transmission window size in Fig. \ref{avwthroughput} (a) and (b), respectively. The average performances of both offline optimal policy and online heuristic tend to saturate as transmission window size increases beyond $100$ slots. The experiment is repeated in Fig. \ref{avwthroughputmemory},for the case where energy harvesting process has a memory remaning in the same state with $0.9$ probability and switching to other state with probability $0.1$.

\begin{figure}[htpb]
\centering
  \begin{psfrags}
    \psfrag{A}[t]{Slot Index}
	\psfrag{B}[b]{Power (in Watts)}
	\psfrag{L}[l]{\tiny{HEURISTIC}}
	\psfrag{C}[l]{\tiny{CHANNEL}}
	\psfrag{P}[l]{\tiny{PACKET ARRIVALS}}
	\psfrag{E}[l]{\tiny{ENERGY ARRIVALS}}
	\psfrag{O}[l]{\tiny{OPTIMAL}}
    \includegraphics[scale=0.27]{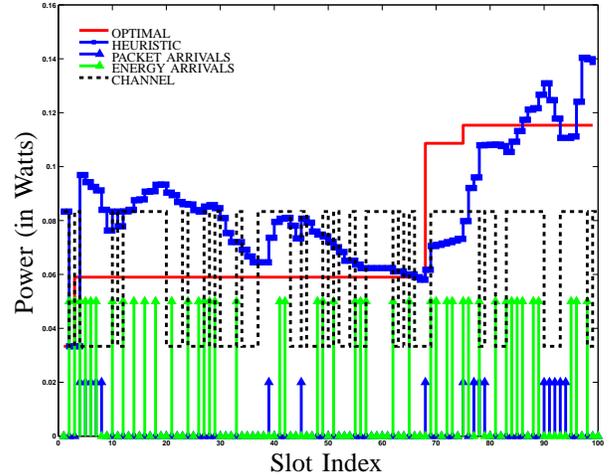}
    \end{psfrags}
\caption{Water level profiles of throughput maximizing optimal offline policy (Red Curve) and online heuristic policy (Blue Curve) for a sample realization of packet arrival, energy harvesting and channel fading processes assuming  a Markovian stream of $10$kB packets (Blue Arrows) having two states ($l(0)=0$ and $l(1)=10$kB) with transition probabilities $q_{00}=0.9$, $q_{01}=0.1$, $q_{10}=0.58$, $q_{11}=0.42$, a Gilbert-Elliot Channel (Dashed Black Curve) where good ($\gamma^{good}=30$) and bad ($\gamma^{bad}=12$) appear with equal probabilies ($P(\gamma_{n}=\gamma^{good})=0.5$ and $P(\gamma_{n}=\gamma^{bad})=0.5$ ) and energy harvests (Green Arrows) of $50$nJs occuring with a probability of $0.5$ at each slot.}
\label{sample12a} 
\end{figure}

\begin{figure}[htpb]
\centering
  \begin{psfrags}
    \psfrag{A}[b]{Power (in mW)}
	\psfrag{B}[t]{Slot Index}
	\psfrag{C}[l]{\tiny{OPTIMAL}}
	\psfrag{D}[l]{\tiny{HEURISTIC}}
\begin{subfigure}[]
    \centering \includegraphics[scale=0.27]{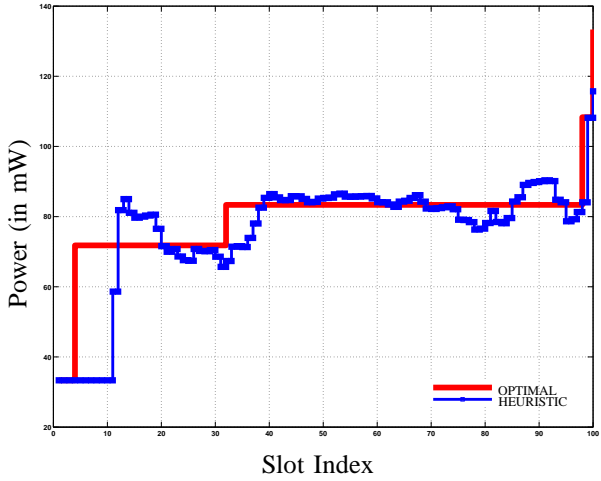}
\end{subfigure}
\begin{subfigure}[]    
    \centering \includegraphics[scale=0.27]{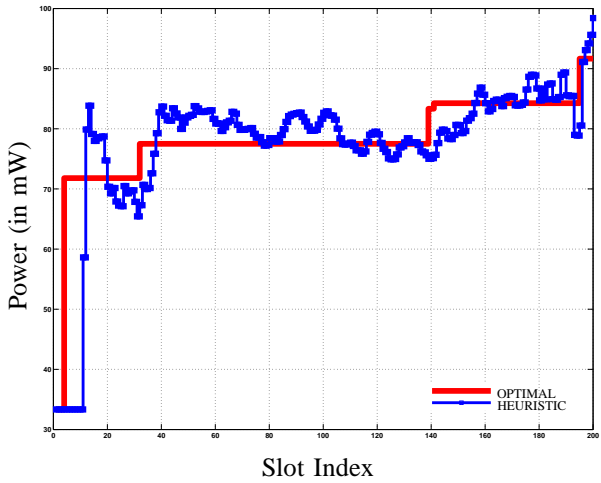}
\end{subfigure}  
    \end{psfrags}  
\caption{Water level profiles of throughput maximizing optimal offline policy (Red Curve) and online heuristic policy (Blue Curve) for a sample realization of packet arrival, energy harvesting and channel fading processes when $N=100$ (a) and $N=200$ (b).}
\label{sample12} 
\end{figure}


\begin{figure}[htpb]
\centering
  \begin{psfrags}
    \psfrag{A}[t]{Slot Index}
	\psfrag{B}[b]{Avg. Throughput (in Mbit/s)}
	\psfrag{P}[b]{Energy Consumption per slot (in nJ)}  
	\psfrag{C}[l]{\tiny{OPTIMAL}}
	\psfrag{D}[l]{\tiny{HEURISTIC}}
\begin{subfigure}[]
    \centering \includegraphics[scale=0.27]{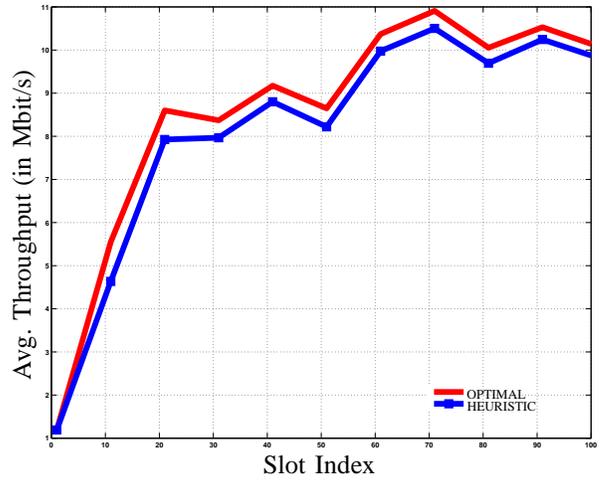}
\end{subfigure}
\begin{subfigure}[] 
     
    \centering \includegraphics[scale=0.27]{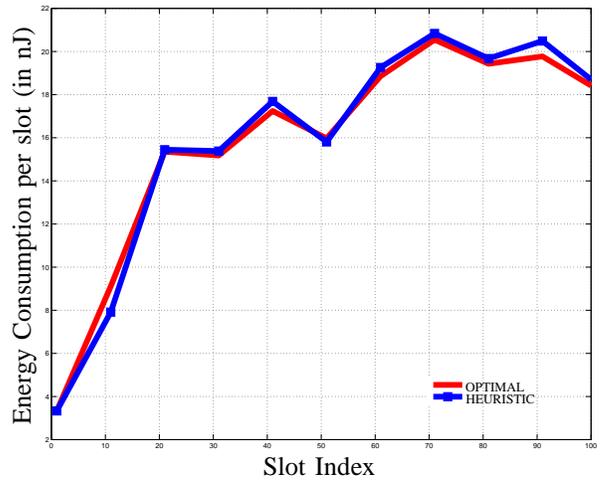}
\end{subfigure}  
    \end{psfrags}  
\caption{Average throughput (a) and energy consumption per slot (b) comparison of throughput maximizing optimal offline policy (red) and online heuristic policy (blue) against varying transmission window size for stationary energy harvesting.}
\label{avwthroughput} 
\end{figure}

\begin{figure}[htpb]
\centering
  \begin{psfrags}
    \psfrag{A}[t]{Slot Index}
	\psfrag{B}[b]{Avg. Throughput (in Mbit/s)}
	\psfrag{P}[b]{Energy Consumption per slot (in nJ)}  
	\psfrag{C}[l]{\tiny{OPTIMAL}}
	\psfrag{D}[l]{\tiny{HEURISTIC}}
\begin{subfigure}[]
    \centering \includegraphics[scale=0.27]{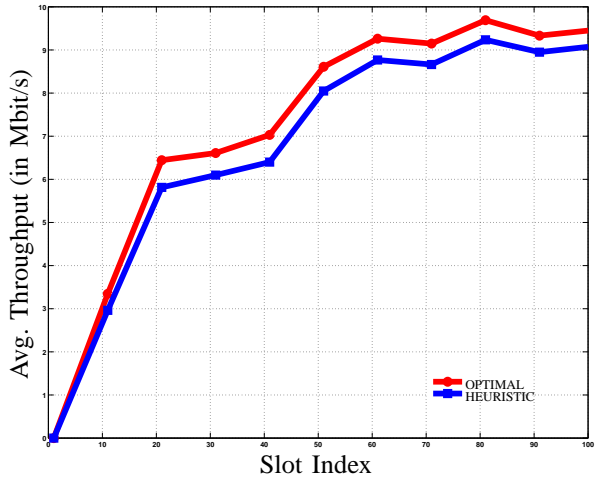}
\end{subfigure}
\begin{subfigure}[] 
     
    \centering \includegraphics[scale=0.27]{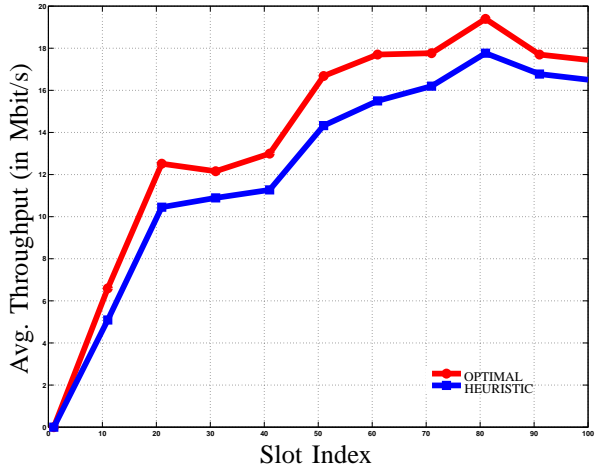}
\end{subfigure}  
    \end{psfrags}  
\caption{Average throughput (a) and energy consumption per slot (b) comparison of throughput maximizing optimal offline policy (red) and online heuristic policy (blue) against varying transmission window size for energy harvesting with memory.}
\label{avwthroughputmemory} 
\end{figure}

\begin{figure}[htpb]
\centering
  \begin{psfrags}
    \psfrag{A}[t]{Sample Path Index}
	\psfrag{B}[b]{Avg. Throughput (in Mbit/s)}
	\psfrag{C}[l]{\tiny{OPTIMAL}}
	\psfrag{D}[l]{\tiny{HEURISTIC}}
    \includegraphics[scale=0.27]{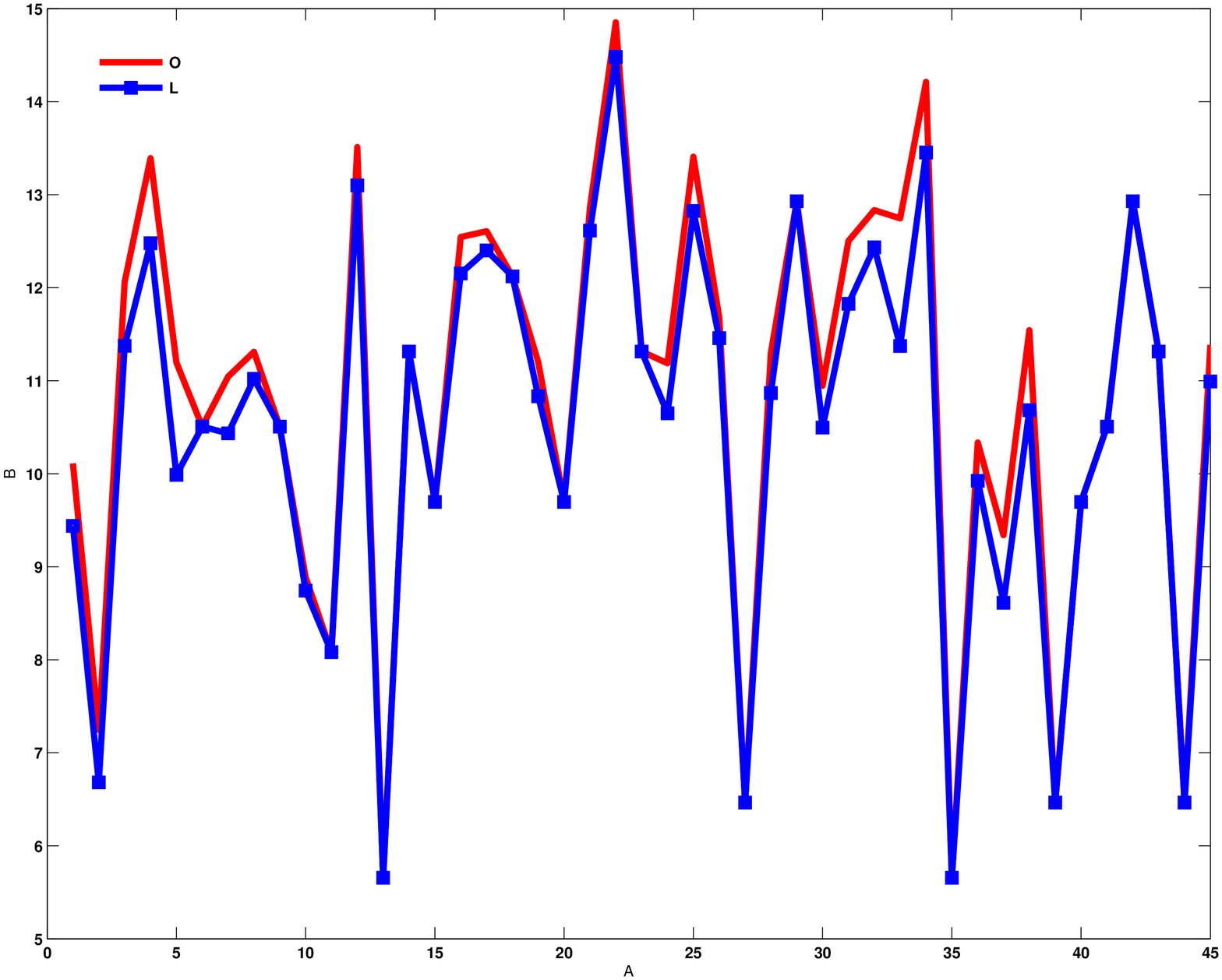}
    \end{psfrags}
\caption{Average Throughput comparison of throughput maximizing optimal offline policy (red) and online heuristic policy (blue) for individual realizations of packet arrival, energy harvesting and channel fading processes assuming  a Markovian stream of $10$kB packets  having two states ($l(0)=0$ and $l(1)=10$kB) with transition probabilities $q_{00}=0.9$, $q_{01}=0.1$, $q_{10}=0.58$, $q_{11}=0.42$, a Gilbert-Elliot Channel where good ($\gamma^{good}=30$) and bad ($\gamma^{bad}=12$) appear with equal probabilies ($P(\gamma_{n}=\gamma^{good})=0.5$ and $P(\gamma_{n}=\gamma^{bad})=0.5$ ) and energy harvests of $50$nJs occuring with a probability of $0.5$ at each slot.}
\label{wthroughput} 
\end{figure}
\section{Conclusion}

In this paper, finite horizon energy efficient transmission schemes are investigated. First, considering only packet arrivals, an online problem of minimizing total energy cost of transmission within a finite horizon and its optimal solution by dynamic programming is posed and expected threshold policy is proposed as a close-to-optimal heuristic. It is also shown by numerical studies that simpler policies which could have sufficient long-term performances can fail in the short term.

Then, a more general problem, considering energy arrivals
as well as channel variation, is defined. The relationship of the
optimal offline solution to throughput maximization is shown
analytically and its optimality for energy cost minimizing
is also proven under certain conditions. Based on offline
throughput maximizing solution, an online heuristic , which
does not require prior statistical knowledge, is presented, and
is observed to achieve close to offline optimal performance in
simulations.

\section*{Acknowledgment}

The authors would like to thank Turk Telekom and TUBITAK (grant no. 110E252) for funding this work.



\bibliographystyle{IEEEtran}
%


\bibliography{LazySched_BlackSeaComm} 


\end{document}